\documentclass[conference]{IEEEtran}
\usepackage{pifont}
\usepackage[T1]{fontenc}
\usepackage{verbatim}
\usepackage{amsfonts}
\usepackage{mathrsfs}
\usepackage{amssymb}
\usepackage{amsmath}
\usepackage[dvips]{graphicx}
\usepackage{float}
\usepackage{url}
\usepackage{stfloats}
\makeatletter

\allowdisplaybreaks[4]

\newtheorem{definition}{Definition}
\newtheorem{proposition}[definition]{Proposition}
\newtheorem{theorem}[definition]{Theorem}

\newtheorem{remark}[definition]{Remark}

\begin{document}
%
% paper title
% can use linebreaks \\ within to get better formatting as desired
\title{Secure Multiplex Coding Over Interference Channel with Confidential Messages}

% author names and affiliations
% use a multiple column layout for up to three different
% affiliations
\author{\IEEEauthorblockN{Xiaolin Li}
\IEEEauthorblockA{Dept.~of Electrical and Computer Engineering\\
Hong Kong University of Science and Technology\\
%Clear Water Bay, Hong Kong\\
darwinlxl@ust.hk } \and \IEEEauthorblockN{Ryutaroh Matsumoto}
\IEEEauthorblockA{Dept.~of Communications and Integrated Systems,\\
Tokyo Institute of Technology,  152--8550 Japan.\\
ryutaroh@it.ss.titech.ac.jp }

%\and
%\IEEEauthorblockN{James Kirk\\ and Montgomery Scott}
%\IEEEauthorblockA{Starfleet Academy\\
%San Francisco, California 96678-2391\\
%Telephone: (800) 555--1212\\
%Fax: (888) 555--1212}
}

% make the title area
\maketitle

\begin{abstract}
In this paper, inner and outer bounds on the capacity region of
two-user interference channels with two confidential messages have
been proposed. By adding secure multiplex coding to the error
correction method in \cite{Han-Konbayashi} which achieves the best
achievable capacity region for interference channel up to now, we
have shown that the improved secure capacity region compared with
\cite{Interference} now is the whole Han-Kobayashi region. In
addition, this construction not only removes the rate loss incurred
by adding dummy messages to achieve security, but also change the
original weak security condition in \cite{Interference} to strong
security. Then the equivocation rate for a collection of secret
messages has also been evaluated, when the length of the message is
finite or the information rate is high, our result provides a good
approximation for bounding the worst case equivocation rate. Our
results can be readily extended to the Gaussian interference channel
with little efforts.
\end{abstract}


%\markboth{Weekly Report on Nov 16th,2009}{Matroid And Non-shannon
%inequality}

\section{Introduction}

Information theoretic security \cite{liang09} attracts a lot of
attention as security is one of the most important issues in
communication, and it guarantees security even when the adversary
has unlimited computing power.

Interference channel \cite{interference_model} has been one of the
most important channel models investigated in information theory as
it captures the main features of the multi-input multi-output
communication system. Signals from different transmitters in this
model exert influence on each other, which also adds the necessity
and difficulty for secure communication.

In the paper \cite{Han-Konbayashi_ori} that proposed the
Han-Kobayashi region which provided the best inner bound known to
now, the information from each transmitter was divided into two
parts, the first part was for only one receiver (we say this part of
the information is sent over the ``private channel''), and the other
part could be decoded by both receivers (we say this part of the
information is sent over the ``common channel''). This naturally
raises one question: what is the secrecy transmission rate if only
confidential messages are sent? In \cite{Interference}, the authors
proposed a scheme which is just a modification of the coding scheme
in \cite{Han-Konbayashi}, but they only sent information on the
``private channel''. This is a natural solution, but in this paper,
we show somewhat surprisingly that even if we transmit over the
``common channel'', confidentiality can also be guaranteed, thus we
propose a larger achievable security rate region.

%In \cite{Interference}, two independent and confidential messages
%are intended to be sent to two receivers respectively through a
%discrete memoryless inference channel. Each message is intended for
%only one receiver and the other receiver is required to be kept as
%ignorant as possible.

Also in \cite{Interference}, outer and inner bounds have been
provided, but are under the weak secrecy requirement \cite{Secrecy},
which requires that the mutual information divided by the length of
the codeword goes to zero as the codeword length goes to infinity.
But this requirement is not strong enough for some applications
\cite{barros08} \cite{maurer94}, because even if this rate goes to
zero asymptotically, vital information bits can still be easily
leaked to an illegitimate receiver. Moreover, secrecy is achieved by
adding dummy random bits into the transmitted signal, which
inevitably decreases the information rate.

The authors in \cite{Interference} did not evaluate the equivocation
rate when the information rates of the secret messages are large or
the length of the message is finite. This means that their results
are only valid for the cases where secrecy can be asymptotically
achieved, but if the secrecy requirement is not achieved, they are
not able to evaluate how much information may be leaked out.

In \cite{strong_inter}, the authors calculated the secure degree of
freedom achievable with strong security requirement in interference
channels. But the degree of freedom is only a crude measure for
information transmission speed, and the knowledge on the capacity
region of the interference channel with strong security requirement
remains to be limited.

In \cite{yamamoto05}, the authors proposed the secure multiplex
coding scheme for wiretap channels, the goal of which is to remove
the rate loss incurred by the random dummy message. The main idea is
to transmit $T$ statistically independent secret messages
simultaneously, and for each secret message, other messages serve as
``random bits'', making it ambiguous for eavesdroppers. In
\cite{matsumotohayashi2011eprint} and
\cite{matsumotohayashi2011netcod}, the authors applied the secure
multiplex coding in different scenarios: broadcast channels with a
common message and secure network coding. They showed that secure
multiplex coding can not only remove the information rate loss, but
can also achieve strong security within the capacity region. Despite
all these findings, it is still not clear whether such technique can
also be generalized to other multiuser communication scenarios.

In this paper, the model of interference channel with confidential
messages as in \cite{Interference} is considered, by applying the
technique of secure multiplex coding,  we have proposed inner and
outer bounds on the capacity region within which the strong security
requirement can be achieved. Moreover, we give the dominating term
approximation for a lower bound on the equivocation rate with finite
message length. We also show that all the above results can be
easily carried over to the Gaussian interference channel case.

This paper is organized as the following: in Section II, the system
model and the necessary mathematical tools used shall be introduced.
In Section III, the random coding scheme is presented, based on
which we propose an inner bound on the capacity region of the
interference channel. An outer bound is also proposed. In Section IV
we extend our results to the Gaussian interference channel. We
provide some discussion and comparison of our results with that in
\cite{Interference} in Section V. Section VI concludes the paper.

\section{System Model and Preliminaries}

\subsection{System Model}
We adopt the same channel model as in \cite{Interference}. Consider
a discrete memoryless interference channel with finite input
alphabets $\mathcal {X}_1,\ \mathcal{X}_2$, finite output alphabets
$\mathcal {Y}_1,\ \mathcal{Y}_2$, and the channel transition
probability distribution $P_{Y_1, Y_2| X_1, X_2}$. Two transmitters
wish to send independent, confidential messages to their respective
receivers. The channel model is illustrated in Figure \ref{Interf}.

\begin{figure}[t]
\centering{}\includegraphics[scale=0.16,trim=0 0 0
0,clip=]{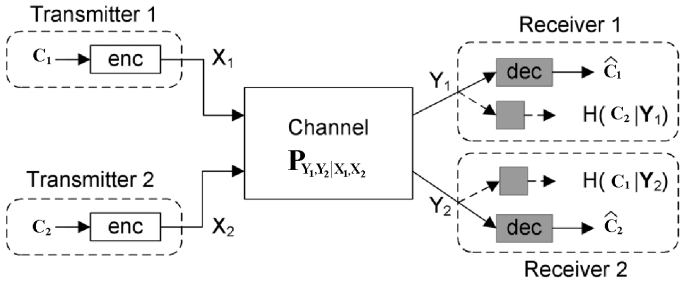}\caption{Interference channel with confidential
messages.} \label{Interf}
\end{figure}

The main goals of communication under this framework are:

\begin{enumerate}
\item To ensure the decoding error probability for each receiver to be
small enough;

\item Secrecy requirement, which means the receiver intending to
receive one message should be kept in ignorance for the other
message.
\end{enumerate}

By secure multiplex coding, we mean that multiple statistically
independent messages are sent over virtually different channels
(actually all these messages are sent simultaneously through the
same physical channel), and any such channel or collection of
channels is required to be secure to the unintended receiver.

To make the above arguments accurate, the definition of the capacity
region for the interference channel with secure multiplex coding is
given as follow:

\begin{definition}
The rate tuple $(R_{1,1},\ \ldots,\ R_{1,T_1},$ $ R_{2,1},\ \ldots,\
R_{2,T_2})$ and the equivocation rate tuple $\{
(R_{1,\mathcal{I}_1,e}, R_{2,\mathcal{I}_2,e}) \mid  \emptyset \neq
\mathcal{I}_1\subseteq \{1, \ldots, T_1\},\  \emptyset
\neq\mathcal{I}_2 \subseteq \{1, \ldots, T_2\}\}$ are said to be
\emph{achievable} for the secure multiplex coding with $T_1$ secret
messages for sender 1 and $T_2$ messages for sender 2, if there
exists a sequence of stochastic encoders for sender 1 denoted as
$\zeta_{1}^{n}$ : $\mathcal{C}_{1,1}^{n} \times \cdots \times
\mathcal{C}_{1,T_1}^{n} \rightarrow \mathcal{X}_{1}^{n}$, and for
sender 2 denoted as $\zeta_{2}^{n}$ : $\mathcal{C}_{2,1}^{n} \times
\cdots \times \mathcal{C}_{2,T_2}^{n} \rightarrow
\mathcal{X}_{2}^{n}$, and for receiver 1 deterministic decoder
$\varphi_{1}^{n}: \mathcal{Y}_{1}^{n} \rightarrow
\mathcal{C}_{1,1}^{n} \times \cdots \times \mathcal{C}_{1,T_1}^{n}
$, and for receiver 2 deterministic decoder $\varphi_{2}^{n}:
\mathcal{Y}_{2}^{n} \rightarrow \mathcal{C}_{2,1}^{n} \times \cdots
\times \mathcal{C}_{2,T_2}^{n} $ such that
\begin{align}
\lim_{n\rightarrow \infty} \mathrm{Pr}[
(C_{1,1}^{n},\ldots, C_{1,T_1}^{n}) \neq \varphi_{1}^{n}(Y_{1}^{n}) &\nonumber\\
\textrm{ or } (C_{2,1}^{n},\ldots, C_{2,T_2}^{n}) \neq
\varphi_{2}^{n}(Y_{2}^{n})]
&= 0,\label{error_prob}\\
\lim_{n\rightarrow\infty} I(C_{1,\mathcal{I}_1}^{n}; Y_{2}^{n}) = 0\
\Bigl(\mbox{if } R_{1,\mathcal{I}_1,e} &= \sum_{i\in\mathcal{I}_1}R_{1,i}\Bigr),\label{strong1}\\
\lim_{n\rightarrow\infty} I(C_{2,\mathcal{I}_2}^{n}; Y_{1}^{n}) = 0\
\Bigl(\mbox{if } R_{2,\mathcal{I}_2,e} &= \sum_{i\in\mathcal{I}_2}R_{2,i}\Bigr),\label{strong2}\\
\liminf_{n\rightarrow\infty} H(C_{1,\mathcal{I}_1}^{n}|  Y_{2}^{n})/n &\geq R_{1,\mathcal{I}_1,e},\label{equi1}\\
\liminf_{n\rightarrow\infty} H(C_{2,\mathcal{I}_2}^{n}|  Y_{1}^{n})/n &\geq R_{2,\mathcal{I}_2,e},\label{equi2}\\
\liminf_{n\rightarrow\infty} \frac{\log| \mathcal{C}_{1,i}^{n}| }{n} &\geq R_{1,i},\label{rage1}\\
\liminf_{n\rightarrow\infty} \frac{\log| \mathcal{C}_{2,j}^{n}| }{n}
&\geq R_{2,j},\label{rate2}
\end{align}
for $i=1$, \ldots, $T_1$ and $j=1$, \ldots, $T_2$ , where $C_{1,
i}^{n}$ and $C_{2, j}^{n}$ represent the $i$-th secret message from
sender 1 and the $j$-th secret message from sender 2 respectively.
All of $C_{1, i}^{n}$ and $C_{2, j}^{n}$ have uniform distribution
on $\mathcal{C}_{1, i}^{n}$ and $\mathcal{C}_{2, j}^{n}$ and are
statistically independent. Both of $C_{1,\mathcal{I}_1,}^{n}$ and
$C_{2,\mathcal{I}_2,}^{n}$ are collections of random variables:
$C_{1,\mathcal{I}_1}^{n} = \{C_{1,i}^{n}\mid  i \in \mathcal{I}_1\}$
and $C_{2,\mathcal{I}_2}^{n} = \{C_{2,j}^{n}\mid  j \in
\mathcal{I}_2\}$. The received signals by the two receivers are
denoted as $Y_{1}^{n}$ and $Y_{2}^{n}$, with the transmitted signals
$\zeta_{1}^{n}(C_{1,1}^{n}$, \ldots, $C_{1,T_1}^{n})$,
$\zeta_{2}^{n}(C_{2,1}^{n}$, \ldots, $C_{2,T_2}^{n})$, and the
channel transition probability $P_{Y_1, Y_2\mid X_1, X_2}$. The
capacity region of the secure multiplex coding is the closure of the
achievable rate tuples.
\end{definition}

\begin{remark}
In the above definition we require the mutual information
$I(C_{1,\mathcal{I}_1}^{n}; Y_{2}^{n})$ and
$I(C_{2,\mathcal{I}_2}^{n}; Y_{1}^{n})$ approaches zero as $n$
approaches infinity when $R_{1,\mathcal{I}_1,e} =
\sum_{i\in\mathcal{I}_1}R_{1,i}$ and $R_{1,\mathcal{I}_1,e} =
\sum_{i\in\mathcal{I}_1}R_{1,i}$, this is the requirement of the
strong secrecy according to \cite{Secrecy}.

The main idea behind the multiplex coding is that more
``constraints'' have been put on the confidential message to remove
the rate loss caused by adding dummy message: instead of sending one
confidential message, multiple independent messages are transmitted,
so instead of making the mutual information between
$(\mathcal{C}_{1,1}^{n}, \cdots , \mathcal{C}_{1,T_1}^{n})$ and
$Y_{2}^{n}$ to be zero, we now only need to ensure
$I(\mathcal{C}_{1,i}^{n}, Y_{2}^{n})$ vanishes, which means each
multiplex channel is secure. For other messages
$(\mathcal{C}_{1,1}^{n}, \cdots, \mathcal{C}_{1,i-1}^{n},
\mathcal{C}_{1,i+1}^{n}, \cdots, \mathcal{C}_{1,T_1}^{n})$, since
they are independent with $\mathcal{C}_{1,i}^{n}$, so they acted as
noise and provide protection for $\mathcal{C}_{1,i}^{n}$.
\end{remark}

%\begin{definition}
%The rate tuple $(R_{1,1},\ \ldots,\ R_{1,T_1},$ $ R_{2,1},\ \ldots,\
%R_{2,T_2})$ and the equivocation rate tuple $\{
%R_{1,\mathcal{I}_1,e}, R_{2,\mathcal{I}_2,e} | \emptyset \neq
%\mathcal{I}_1\subseteq \{1, \ldots, T_1\},\  \emptyset
%\neq\mathcal{I}_2 \subseteq \{1, \ldots, T_2\}\}$ are said to be an
%\emph{unachievable} for the secure multiplex coding with $T_1$
%secret messages for sender 1 and $T_2$ messages for sender 2, if
%outside this region specified, no matter what are the structure of
%the transmitters and receivers, \eqref{error_prob} ~ \eqref{rate2}
%can never be satisfied simultaneously.
%\end{definition}

\subsection{Preliminaries}

%\subsubsection{Privacy Amplification Theorem and math tools}
In this paper the main tools we are going to use is the strengthened
privacy amplification theorem, which will be sensitive to the change
of bases. So throughout the whole paper we just use natural log.

\begin{definition}
\cite{carter79} Let $\mathcal{F}$ be a set of functions from
$\mathcal{S}_1$ to $\mathcal{S}_2$, and $F$ the not necessarily
uniform random variable on $\mathcal{F}$. If for any $x_1 \neq x_2
\in \mathcal{S}_1$ we have
\[
\mathrm{Pr}[F(x_1)=F(x_2)] \leq \frac{1}{|\mathcal{S}_2|},
\]
then $\mathcal{F}$ is said to be a \emph{family of two-universal
hash functions}.
\end{definition}
%We shall use a family of two-universal hash functions in which each function is surjective.
%The set of all surjective linear functions is an example of such a family.

\begin{theorem}\label{thm:pa}\cite{multiplex}\cite{matsumotohayashi2011netcod}
Let $L$ be a random variable with uniform distribution over a finite
alphabet $\mathcal{L}$ and $Z$ be any discrete random variable.
%\footnote{We do not
%assume the existence of its probability mass function nor
%probability density function.}
Let $\mathcal{F}$ be a family of two-universal hash functions from
$\mathcal{L}$ to $\mathcal{M}$, and $F$ be a random variable on
$\mathcal{F}$ statistically independent of $L$. Then
\begin{align}
&\mathbf{E}_f \exp(\rho I(F(L);Z|F=f))  \nonumber\\ \leq 1+
%|\mathcal{M}|^\rho\mathbf{E}[P_{L|Z}(L|Z)^\rho]\label{hpa1}
&\frac{|\mathcal{M}|^\rho}{|\mathcal{L}|^\rho} \sum_{z,\ell}
P_L(\ell) P_{Z|L}(z|\ell)^{1+\rho} P_Z(z)^{-\rho}.
\label{hpa1discrete}
\end{align}
for $0<\rho\leq 1$. %If $Z$ is not discrete RV, $I(F(L);Z|F)$ is
%defined to be $H(F(L)|F) - \mathbf{E}_z H(F(L)|F,Z=z)$.

%In addition to the above assumptions, when $L$ is uniformly
%distributed, we have
%\begin{equation}
%|\mathcal{M}|^\rho\mathbf{E}[P_{L|Z}(L|Z)^\rho]=
%\frac{|\mathcal{M}|^\rho\mathbf{E}[P_{L|Z}(L|Z)^\rho
%P_L(L)^{-\rho}]}{|\mathcal{L}|^\rho}.\label{hpa1uni}
%\end{equation}
%In addition to all of the above assumptions, when $Z$ is a discrete
%random variable, we have
%\begin{align}
%&\frac{|\mathcal{M}|^\rho\mathbf{E}[P_{L|Z}(L|Z)^\rho
%P_L(L)^{-\rho}]}{|\mathcal{L}|^\rho}\nonumber\\
%=&\frac{|\mathcal{M}|^\rho}{|\mathcal{L}|^\rho} \sum_{z,\ell}
%P_L(\ell) P_{Z|L}(z|\ell)^{1+\rho} P_Z(z)^{-\rho}.
%\label{hpa1discrete}
%\end{align}
\end{theorem}

\begin{remark}
It was assumed that $Z$ was discrete in
\cite{matsumotohayashi2011netcod}. However, when the alphabet of $L$
is finite, there is no difficulty to extend the original result.
\end{remark}

\begin{definition}
\begin{align}
\psi(\rho, P_{Z|L}, P_L) &= \log \sum_z \sum_\ell P_L(\ell)
P_{Z|L}(z|\ell)^{1+\rho} P_Z(z)^{-\rho},
\label{eq:psid}\\
\phi(\rho,P_{Z|L},P_L) &= \log \sum_z\left( \sum_{\ell} P_{L}(\ell)
(P_{Z|L}(z|\ell)^{1/(1-\rho)})\right)^{1-\rho}. \label{phid}
\end{align}
\end{definition}
Observe that $\phi$ is essentially Gallager's function $E_0$
\cite{gallager68}. The main reason we introduce this function is
that its concavity greatly facilitates the process of derivation.

\begin{proposition}\cite{gallager68,hayashi11}
$\exp(\phi(\rho, P_{Z|L}, P_L))$ is concave with respect to $P_L$
with fixed $0<\rho< 1$ and $P_{Z|L}$. For fixed $0<\rho< 1$, $P_L$
and $P_{Z|L}$ we have

\begin{equation}\label{psileqphi}
\exp(\psi(\rho, P_{Z|L}, P_L))\leq \exp(\phi(\rho, P_{Z|L}, P_L)).
\end{equation}

\end{proposition}

It can be found in \cite{gallager68} that the derivative of the
Gallager's function has a simple expression when $\rho = 0$:

\begin{equation}\label{derivative}
\lim_{\rho \rightarrow 0}\frac{d\phi(\rho, P_{Z|L}, P_L)}{d\rho} =
\sum_{l, z}P_{L, Z}(l, z)\log\frac{P_{Z| L}}{P_{Z}} = I(Z, L).
\end{equation}

%\subsubsection{The Han-Kobayashi region}
Introduction to the Han-Kobayashi region can be found in Lemma 4 in
\cite{Han-Konbayashi}, and is presented as below:

\begin{theorem}

Let $\mathcal{P}^{*}_{1}$ be the set of probability distribution
$P^{*}_{1}(\cdot)$ that factor as

\begin{equation} \label{HKP}
\begin{split}
&P^{*}(u, w_1, w_2, v_1, v_2)\\
%&= p_{U}(u)p_{V_1W_1|U}(v_1,w_1|u)p_{V_2W_2|U}(v_2,w_2|u)
&= p(u)p(v_1,w_1|u)p(v_2,w_2|u)
 \end{split}
 \end{equation}

Let $R_{HK}(P^{*}_{1})$ be the set of nonnegative rate-tuples $(R_1,
R_2)$ that satisfy

\addtocounter{equation}{1}
\begin{align}
R_1&\leq I(V_1;Y_1|W_2U) \label{Han-Kobayashi_1}\\
R_1&\leq I(V_1;Y_1|W_1W_2U)+I(V_2W_1;Y_2|W_2U)\\
R_2&\leq I(V_2;Y_2|W_1U)\\
R_2&\leq I(V_2;Y_2|W_2W_1U)+I(V_1W_2;Y_1|W_1U)\\
R_1+R_2&\leq I(V_1W_2;Y_1|U)+ I(V_2;Y_2|W_1W_2U)\\
R_1+R_2&\leq I(V_1;Y_1|W_1W_2U)+ I(V_2W_1;Y_2|U)\\
R_1+R_2&\leq I(V_1W_2;Y_1|W_1U)+ I(V_2W_1;Y_2|W_2U)\\
2R_1+R_2&\leq I(V_1W_2;Y_1|U)+I(V_1;Y_1|W_1W_2U) +\nonumber\\
&\qquad I(V_2W_1;Y_2|W_2U)\\
R_1+2R_2&\leq I(V_2W_1;Y_2|U)+I(V_2;Y_2|W_1W_2U) + \nonumber\\
&\qquad I(V_1W_2;Y_1|W_1U) \label{Han-Kobayashi_final}\\
R_1,\ R_2 &\geq 0
\end{align}

\noindent Then we have

$$R_{HK} = \cup_{P_{1}^{*}\in \mathcal{P}_{1}^{*}}R_{HK}(P_{1}^{*})$$

\noindent is an achievable rate region for the discrete memoryless
IC.

\end{theorem}

\section{Capacity Region of the Secure Multiplex Coding with Strong Secrecy Requirement}

%In this section, we shall propose inner and outer bounds on the
%capacity region of the secure multiplex coding over interference
%channel with the requirement of strong secrecy.

\subsection{Inner Bound}

Denote the total rate of the sender $t$ by $0\leq R_t =
\sum_{i=1}^{T_t+1} R_{t,i} \leq I(V_t;Y_t| U)$ (here $t = 1$ or 2,
and we adopt this notation throughout the paper). An inner bound is
proposed as the following:

\begin{theorem}\label{inner}

Let $\mathcal{P}^{*}_{2}$ be the set of probability distribution
$P^{*}_{2}(\cdot)$ that factor as

\begin{equation} \label{PIB}
\begin{split}
&P(u, w_1, w_2, v_1, v_2, x_1, x_2, y_1, y_2)\\
&= P(u)P(w_1, w_2, v_1, v_2, x_1, x_2|u)P(y_1,y_2|x_1, x_2)\\
&= P(u)P(w_1, v_1|u)P(x_1|v_1)P(w_2,
v_2|u)P(x_2|v_2)\\
&\qquad P(y_1,y_2|x_1, x_2)
 \end{split}
 \end{equation}

\noindent Here $x_1,\ x_2$ and $y_1, \ y_2$ are inputs and outputs
for the interference channel respectively.

And $R_{in}(P_{2}^{*})$ be the set of nonnegative rate-tuples and
$(R_1, R_2, R_{1,\mathcal{I}_1, e}, R_{2,\mathcal{I}_2, e})$ satisfy

\begin{align}
R_{1,\mathcal{I}_1, e}^{'} &= R_{1,\mathcal{I}_1, e} + I(V_1;Y_2|U,
V_2)\label{def1}\\
R_{2,\mathcal{I}_2, e}^{'} &= R_{2,\mathcal{I}_2, e} + I(V_2;Y_1|U,
V_1)\label{def2}\\
(R_1, R_2) &\in R_{HK}(P_{2}^{*}) \label{ratepair}\\
(R_{1,\mathcal{I}_1, e}^{'},
R_{2,\mathcal{I}_2, e}^{'}) &\in R_{HK}(P_{2}^{*})\label{eqvopair}\\
0 \leq R_{1,\mathcal{I}_1, e}&\leq \sum_{i\in\mathcal{I}_1}R_{1,i}\label{collrate1}\\
0 \leq R_{2,\mathcal{I}_2, e}&\leq \sum_{j\in\mathcal{I}_2}R_{2,j}\label{collrate2}\\
\end{align}

\noindent Note in the above we abuse the notation a little by
writing $R_{HK}(P_{2}^{*})$, we can write this because if $P(u, w_1,
w_2, v_1, v_2, x_1, x_2, y_1, y_2)\in P_{2}^{*}$, then the marginal
distribution $P(u, w_1, w_2, v_1, v_2)\in P_{1}^{*}$.

An inner bound for the interference channels with secure multiplex
coding is

$$R_{in} = \cup_{P^{*}_{2}\in \mathcal{P}^{*}_{2}}R_{in}(P_{2}^{*})$$

\end{theorem}

\begin{remark}
The inner bound of secret capacity over interference channel given
above shows that the whole Han-Kobayashi region can be achieved,
which means that in our proposed coding method, the channel capacity
of the interference channel has been fully utilized, and is
guaranteed to be secure.

From \eqref{eq:gotozero} to \eqref{eq:equi2}, we can see that when
\eqref{def1} -- \eqref{collrate2} are satisfied, then the strong
security can be achieved. Note that \eqref{eq:gotozero} --
\eqref{eq1000} also provides an upper bound for the leaked
information, which is not analyzed in \cite{Interference}.

\end{remark}

\begin{proof}
To prove that the above region is an inner bound on the capacity
region, we need to explicitly show that there does exist certain
scheme that can achieve the bound. In part 1 of the proof, we
present the scheme, and in part 2 we evaluate the equivocation rate.

\textbf{Part 1}: \emph{Random Coding Scheme}

Before we present the random coding scheme, some notations are
introduced here: let $(c_{t,1}^{n}, \cdots, c_{t,T_t}^{n})\in
(\mathcal{C}_{t,1}^{n}, \cdots, \mathcal{C}_{t,T_t}^{n})$ be the
secret messages for transmitter $t$, and denote
$\mathcal{C}_{t}^{n}=\prod_{i=1}^{T_t+1}\mathcal{C}_{t,i}^{n}$,
where $\mathcal{C}_{t,T_t+1}^{n}$ is the alphabet of randomness used
by the stochastic encoder, and $n$ denotes the code length. In here
and all the following expressions $t = 1$ or $2$. Let
$\mathcal{F}_{t}^{n}$ be the set of all linear bijective maps from
$\mathcal{C}_{t}^{n}$ to itself.

We modify the random coding scheme proposed in
\cite{Han-Konbayashi}, and apply the secure multiplex coding
techniques. The new scheme is described in detail as the follows:

Fix the distribution of $P(u)$, $P(w_t, v_t|u)$  and $P(x_t|v_t)$,
also since the channel distribution $P(y_1, y_2| x_1, x_2)$ is
given, all the distributions in \eqref{PIB} are now fixed.

\emph{1. Codebook Generation}: Sender $t$ and receiver  $t$ fix and
agree on the choice of a bijective function $f_{t}^{n}\in
\mathcal{F}_{t}^{n}$. Given $T_t$ secret messages $(c_{t,1}^{n},
\cdots, c_{t,T_t}^{n})\in (\mathcal{C}_{t,1}^{n}, \cdots,
\mathcal{C}_{t,T_t}^{n})$, uniformly choose $c_{t,T_t+1}^{n}\in
\mathcal{C}_{t,T_t+1}^{n}$, let $c_{t} =
(f_{t}^{n})^{-1}(c_{t,1}^{n}, \cdots, c_{t,T_t+1}^{n})$. Here the
message $c_{t,T_t+1}^{n}$ is used by the stochastic encoder to
increase the randomness in the secret message.

In order to use multiplex coding, we write:

\begin{align}\label{C}
C_{t}^{n} = (E_{t}^{n}, B_{t}^{n}) = (f_{t}^{n})^{-1}(C_{t,1}^{n},
\cdots, C_{t, T_t+1}^{n})
\end{align}

In \eqref{C}, $f_{t}^{n}$ belongs to the family of linear bijective
maps $F_{t}^{n}$, and this is achieved by matrix multiplication.
Applying $(f_{t}^{n})^{-1}$ on the secret messages can be achieved
in the following way:

$$C_{t}^{n} = (L_t)^{-1}* [C_{t,1}^{n}, \cdots, C_{t, T_t+1}^{n}]^T$$

Note that if the length of $C_{t,i}^{n}$ is $k_{t,i}$ bits, then
$L_t$ is a nonsingular matrix of size $l_t\times l_t$ with
$l_t=\sum_{1\leq i\leq T_t+1} k_{t,i}$. Since $C_{t}^{n}$ has $l_t$
bits, we just need to take some part of the bits for $B_{t}^{n}$ and
part for $E_{t}^{n}$, then the condition of independence will be
satisfied. This is guaranteed by the uniformness of $L_t$ and
$C_{t}^{n}$.

Equation \eqref{C} actually means that we do not distinguish which
part was to be sent over the ``private channel'' and which was to be
sent over the ``common channel'', after the random bijective
mapping, we just divide the message into two parts and sent them.
But we need to require that $E_{t}^{n}$ and $B_{t}^{n}$ are mutually
independent, this can be achieved because all the messages
$c_{t,i}^{n}$ have uniform distribution over its alphabet and are
all independent.

Then in the following, we will encode $E_{t}^{n}$ and $B_{t}^{n}$ in
two different ways.

Randomly generate a sequence $\textbf{u}$ with probability
$P(\textbf{u}) = \prod_{i=1}^{n}P(u_i)$, and assume that both
transmitters and receivers know the time-sharing sequence
$\textbf{u}$.

For transmitter $t$, generate $2^{nS_t}$ independent sequences
($S_t$ is the information rate over the ``common channel'')
$\textbf{w}_t$ each with probability $P(\textbf{w}_t| \textbf{u}) =
\prod_{i=1}^{n}P(w_{t,i}| u_i)$. Then generate $2^{nR_t}$ ($R_t =
S_t + T_t$, and $T_t$ is the rate of information over ``private
channel'') independent sequences $\textbf{v}_t$ each with
probability $P(\textbf{v}_t|\textbf{w}_t \textbf{u}) =
\prod_{i=1}^{n}P(v_{t,i}|w_{t,i}, u_i)$.
%
%The codebook thus generated is denoted as $\lambda_t$, obviously
%$\lambda_t$ contains in total $M_t = 2^{nR_t}$ different codewords.
%(Without loss of generality we can assume $M_t$ to be an integer.)
%Denote $\lambda = (\lambda_1, \lambda_2)$.

\emph{2. Encoding}: Encode $c_t$ with encoder $t1$ and $t2$ (here $t
= 1$ or 2, refer to Fig. \ref{fig2}), according to the codebook
generated in the previous step and obtain the codeword $v_{t}^{n}$.
Then the transmitters generate the channel input sequences based on
respective mappings $P_{X_1| V_1}$ and $P_{X_2| V_2}$. Actually this
step is to apply artificial noise to $v_{t}^{n}$ according to the
conditional probability distribution and get the transmitted signal.
This step is to make the channel of the other receiver more noisy,
and the intended receiver $t$ is supposed to know this $P_{X_t|
V_t}$.

The encoder structure is illustrated by the following figure:

\begin{figure}[H]
\centering{}\includegraphics[scale=0.18,trim=0 0 0
0,clip=]{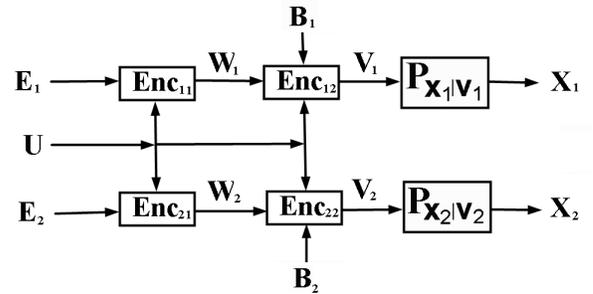}\caption{Code construction for the interference
channel with confidential messages.} \label{fig2}
\end{figure}

In the above figure, $U$ is the time-sharing sequence known by all
the transmitters and receivers. The codebook for encoder $ij$ is
denoted as $\lambda_{ij}\in \Lambda_{ij}$ ($i, j = 1,2$), for the
simplicity of notations, let $\lambda =
(\lambda_{11},\lambda_{12},\lambda_{21},\lambda_{22})$, and $\Lambda
= (\Lambda_{11},\Lambda_{12},\Lambda_{21},\Lambda_{22})$.

According to the structure of the encoder, we have the following
relationship among the variables:

$$W_{t}^{n} = \Lambda(E_{t}^{n}, U^n),\ V_{t}^{n} =
\Lambda(B_{t}^{n}, W_{t}^{n}, U^{n})$$

\emph{3. Decoding}: Without loss of generality, we consider for
receiver 1. Let $A_{\epsilon}^{(n)}(W_1, V_1, W_2, Y_1, U)$ denote
the set of jointly typical sequences defined in \cite{T_Cover_info}
page 521. Since receiver 1 is supposed to know the time sharing
sequence $U^n$, after receiving $Y_{1}^{n}$, the receiver will try
to find $\hat{W_{1}^{n}}, \hat{V_{1}^{n}}, \hat{W_{2}^{n}}$ such
that $(\hat{W_{1}^{n}}, \hat{V_{1}^{n}}, \hat{W_{2}^{n}}, Y_{1}^{n},
U^n)\in A_{\epsilon}^{(n)}(W_1, V_1, W_2, Y_1, U)$. When such choice
exists and is unique then the decoding is successful, then message
$E_{1}^{n}$ and $B_{1}^{n}$ can be restored, and the confidential
message can be readily obtained by $f_{1}^{n}(E_{1}^{n},
B_{1}^{n})$; otherwise declare error.

\textbf{Part 2}:\emph{ Evaluation of Equivocation Rate}

If we compare our modified scheme with the original one in
\cite{Han-Konbayashi}, it can be easily observed that the decoding
error probability in the new scheme is at least as good as the
original one, which saves us the efforts to analyze the probability
of decoding error. Hence we only need to show the existence of the
tuple $(u,\lambda, f_{1}^{n}, f_{2}^{n} )$ such that the strong
secrecy requirement \eqref{strong1}, \eqref{strong2}, \eqref{equi1}
and \eqref{equi2} can be fulfilled. Without loss of generality, we
just need to consider the information leaked to receiver 2 from
sender 1. In \cite{multiplex}, it has been proved that if
$F_{1}^{n}$ is an uniform random variable on $\mathcal{F}_{1}^{n}$
and $\alpha_{1,\mathcal{I}}$ is the projection from
$\mathcal{C}_{1}^{n}$ to
$\prod_{i\in\mathcal{I}}\mathcal{C}_{1,i}^{n}$, then
$\alpha_{1,\mathcal{I}} \circ F_{1}^{n}$ is a family of
two-universal hashing functions. The projection is simply
$\alpha_{1,\mathcal{I}}(C_{1}^{n})\triangleq \{C_{i,1}^{n}\mid i\in
\mathcal{I}\}$. With a little abuse of notations, we also write
$C_{1}^{n} = (F_{1}^{n})^{-1}(C_{1,1}^{n}, \ldots, C_{1,T_1+1}^{n})$
(but the distinction can be made between the context compared with
equation \eqref{C}), by the uniformness of the distribution it can
be seen $C_{1}^{n}$ and
$F_{1}^{n}$ are statistically independent. %Let $\Lambda_t$ be an
%random variable indicating selection of codebook in the random
%ensemble constructed in part 1 for encoder $t$, for the simplicity
%of notation we use $\Lambda = (\Lambda_1, \Lambda_2)$.

We first calculate the mutual information
$I(\alpha_{1,\mathcal{I}}(F_{1}^{n}(C_{1}^{n}));Y_{2}^{n}|F_{1}^{n},\Lambda,
U^n)$ averaged over all possible choices of $(u,\lambda, f_{1}^{n}
)$. Then by some probabilistic argument, the achievability of the
theorem can be proved.

The following derivation is similar to
\cite{matsumotohayashi2011eprint} and \cite{multiplex}, we first fix
the code book $\lambda$ and the synchronization sequence $u$, and
$\rho$ is a real constant with $0< \rho <1$.

\begin{align}
&\mathbf{E}_{f_{1}^{n}}\exp\Big(\rho I\left(\alpha_{1,\mathcal{I}}\left(F_{1}^{n}\left(C_{1}^{n}\right)\right);Y_{2}^{n}|F_{1}^{n}=f_{1}^{n},\Lambda=\lambda, U^n = u\right)\Big)\nonumber\\
&\leq\mathbf{E}_{f_{1}^{n}} \exp\Big(\rho I(\alpha_{1,\mathcal{I}}(F_{1}^{n}(C_{1}^{n}));Y_{2}^{n},C_{2}^{n}|F_{1}^{n}=f_{1}^{n},\nonumber\\
&\qquad \Lambda=\lambda, U^n = u)\Big)\nonumber\\
%&\textrm{(Giving the message $B_{2}^{n}$ does not increase $I$ much.)}\nonumber\\
 &= \mathbf{E}_{f_{1}^{n}} \exp\Big(\rho  \sum_{c_2}P_{C_{2}^{n}}(c_2)I(\alpha_{1,\mathcal{I}}(F_{1}^{n}(C_{1}^{n}));Y_{2}^{n}|\nonumber\\
 &\qquad F_{1}^{n}=f_{1}^{n}, C_{2}^{n}=c_2, \Lambda=\lambda, U^n = u)\Big)\nonumber\\
&\leq\mathbf{E}_{f_{1}^{n}} \sum_{b_2, e_2}P_{B_{2}^{n}}(b_2)P_{E_{2}^{n}}(e_2)\exp\Big(\rho I(\alpha_{1,\mathcal{I}}(F_{1}^{n}(C_{1}^{n}));Y_{2}^{n}| \nonumber\\
&\qquad F_{1}^{n}=f_{1}^{n}, B_{2}^{n}=b_2, E_{2}^{n} = e_2,\Lambda=\lambda, U^n = u)\Big)\nonumber\\
&\leq 1+ \frac{\exp(n\rho R_{\mathcal{I}})}{\exp(n\rho R_p)}\sum_{b_2, e_2}P_{B_{2}^{n}}(b_2)P_{E_{2}^{n}}(e_2) \sum_{c_1,y_2}P_{C_{1}^{n}}(c_1)\nonumber\\
&\qquad P_{Y_{2}^{n}|C_{1}^{n},B_{2}^{n}, E_{2}^{n},\Lambda=\lambda, U^n = u}(y_2|c_1,b_2, e_2)^{1+\rho}\nonumber\\
&\qquad P_{Y_{2}^{n}|B_{2}^{n}=b_2, E_{2}^{n}=e_2,\Lambda=\lambda,
U^n = u}(y_2)^{-\rho} \\
&\qquad\textrm{(by Eq. \eqref{hpa1discrete})}\label{conti}
\end{align}

where

\begin{eqnarray}
R_{\mathcal{I}} &=& \frac{\sum_{i\in\mathcal{I}}
\log|\mathcal{C}_{1,i}^{n}|}{n},\\ \label{eq:riconstraint} R_p &=&
\frac{ \log|\mathcal{C}_{1}^{n}|}{n}.\label{eq:rp}
\end{eqnarray}

We have the following relationship: $W_{t}^{n} = \Lambda(E_{t}^{n},
U^n)$ and $V_{t}^{n} = \Lambda(B_{t}^{n}, W_{t}^{n}, U^{n})$, thus

\begin{align}
&\sum_{b_2, e_2}P_{B_{2}^{n}}(b_2)P_{E_{2}^{n}}(e_2)
\sum_{c_1,y_2}P_{C_{1}^{n}}(c_1)\nonumber\\
&\qquad P_{Y_{2}^{n}|C_{1}^{n},B_{2}^{n}, E_{2}^{n},\Lambda=\lambda, U^n = u}(y_2|c_1,b_2, e_2)^{1+\rho}\nonumber\\
&\qquad P_{Y_{2}^{n}|B_{2}^{n}=b_2, E_{2}^{n}=e_2,\Lambda=\lambda,
U^n = u}(y_2)^{-\rho}\nonumber\\
&=\sum_{y_2,v_1,v_2,w_1, w_2}\sum_{b_2:(\lambda, b_2, w_2,
u)=v_2}P_{B_{2}^{n}}(b_2)\sum_{w_2:(\lambda, e_2, u) =
w_2}P_{E_{2}^{n}}(e_2)\nonumber\\
&\qquad \sum_{b_1:(\lambda, b_1, w_1,
u)=v_1}P_{B_{1}^{n}}(b_1)\sum_{w_1:(\lambda, e_1, u) =
w_1}P_{E_{1}^{n}}(e_1)\nonumber\\
&\qquad P_{Y_{2}^{n}|B_{1}^{n},B_{2}^{n}, W_{1}^{n},
W_{2}^{n},\Lambda=\lambda, U^n =
u}(y_2|b_1,b_2,w_1, w_2)^{1+\rho}\nonumber\\
&\qquad P_{Y_{2}^{n}|B_{2}^{n}=b_2, W_{2}^{n}=w_2, \Lambda=\lambda,
U^n =
u}(y_2)^{-\rho}\nonumber\\
&= \sum_{v_1, v_2,w_1, w_2, y_2}P_{V_{2}^{n}| \Lambda = \lambda, W_{2}^{n} = w_2, U^n=u}(v_2)P_{W_{2}^{n}|\Lambda=\lambda, U^n=u}(w_2)\nonumber\\
&\qquad P_{V_{1}^{n}| \Lambda = \lambda, W_{1}^{n} = w_1,
U^n=u}(v_1)P_{W_{1}^{n}|\Lambda=\lambda, U^n=u}(w_1) \nonumber\\
&\qquad P_{Y_{2}^{n}|V_{1}^{n},V_{2}^{n},\Lambda=\lambda, U^n =
u}(y_2|v_1,v_2)^{1+\rho}\nonumber\\
&\qquad P_{Y_{2}^{n}|V_{2}^{n}=v_2, \Lambda=\lambda, U^n =
u}(y_2)^{-\rho}\nonumber\\
&= \sum_{v_1, v_2, y_2}P_{V_{2}^{n}| \Lambda = \lambda, U^n=u}(v_2)
P_{V_{1}^{n}| \Lambda = \lambda,
U^n=u}(v_1)\nonumber\\
&\qquad P_{Y_{2}^{n}|V_{1}^{n},V_{2}^{n},\Lambda=\lambda, U^n =
u}(y_2|v_1,v_2)^{1+\rho}\nonumber\\
&\qquad P_{Y_{2}^{n}|V_{2}^{n}=v_2, \Lambda=\lambda, U^n =
u}(y_2)^{-\rho}\label{tov1}
\end{align}

Take \eqref{tov1} into \eqref{conti} and continue the derivation, we
have

\begin{align}
&\mathbf{E}_{f_{1}^{n}}\exp\Big(\rho I(\alpha_{1,\mathcal{I}}(F_{1}^{n}(C_{1}^{n}));Y_{2}^{n}|F_{1}^{n}=f_{1}^{n},\Lambda=\lambda, U^n = u)\Big)\nonumber\\
%&\leq 1+ \sum_{v_2}P_{V_{2}^{n}| \Lambda = \lambda, U^n=u}(v_2)\frac{\exp(n\rho R_{\mathcal{I}})}{\exp(n\rho R_p)} \nonumber\\
%& \sum_{v_1,y_2}P_{V_{1}^{n}| \Lambda = \lambda, U^n=u}(v_1)P_{Y_{2}^{n}|V_{1}^{n},V_{2}^{n},\Lambda=\lambda, U^n = u}(y_2|v_1,v_2)^{1+\rho}\nonumber\\
%&\qquad P_{Y_{2}^{n}|V_{2}^{n}=v_2,\Lambda=\lambda, U^n = u}(y_2)^{-\rho}\nonumber\\
&\leq 1+\sum_{v_2}P_{V_{2}^{n}| \Lambda = \lambda, U^n=u}(v_2)\exp\Big(n\rho(R_{\mathcal{I}}-R_p) +\nonumber\\
&\qquad \psi(\rho,P_{Y_{2}^{n}|V_{1}^{n},V_{2}^{n}=v_2,\Lambda=\lambda, U^n = u},P_{V_{1}^{n}|V_{2}^{n}=v_2,\Lambda=\lambda, U^n = u})\Big)\nonumber\\
&\qquad\textrm{ (by \cite{matsumotohayashi2011eprint} and Eq.\ (\ref{eq:psid}))},\nonumber\\
&\leq 1+\sum_{v_2}P_{V_{2}^{n}| \Lambda = \lambda, U^n=u}(v_2)\exp\Big(n\rho(R_{\mathcal{I}}-R_p) +\nonumber\\
&\qquad \phi(\rho,P_{Y_{2}^{n}|V_{1}^{n},V_{2}^{n}=v_2,\Lambda=\lambda, U^n = u},P_{V_{1}^{n}|V_{2}^{n}=v_2,\Lambda=\lambda, U^n = u})\Big)\nonumber\\
&\qquad \textrm{ (by Eq.\ (\ref{psileqphi}))}\nonumber
\end{align}

Then we average the above upper bound over $\Lambda$ and $U^n$:

\begin{align}
&\exp\Big(\rho I(\alpha_{1,\mathcal{I}}(F_{1}^{n}(C_{1}^{n}));Y_{2}^{n}|F_{1}^{n},\Lambda, U^n)\Big)\nonumber\\
&=\exp\Big(\mathbf{E}_{f_{1}^{n},\lambda, u}\rho I(\alpha_{1,\mathcal{I}}(F_{1}^{n}(C_{1}^{n}));Y_{2}^{n}|F_{1}^{n}=f_{1}^{n},\nonumber\\
&\qquad \Lambda=\lambda, U^n = u)\Big)\nonumber\\
&\leq\mathbf{E}_{f_{1}^{n},\lambda, u}\exp\Big(\rho I(\alpha_{1,\mathcal{I}}(F_{1}^{n}(C_{1}^{n}));Y_{2}^{n}|F_{1}^{n}=f_{1}^{n},\nonumber\\
&\qquad \Lambda=\lambda, U^n = u)\Big)\nonumber\\
&\leq 1+\mathbf{E}_{v_2,\lambda, u}\exp\Big(n\rho(R_{\mathcal{I}}-R_p) +\nonumber\\
&\qquad \phi(\rho,P_{Y_{2}^{n}|V_{1}^{n},V_{2}^{n}=v_2,\Lambda=\lambda, U^n = u},P_{V_{1}^{n}|V_{2}^{n}=v_2,\Lambda=\lambda, U^n = u})\Big)\nonumber\\
&\leq 1+\mathbf{E}_{v_2, u}\exp\Big(n\rho(R_{\mathcal{I}}-R_p) +\nonumber\\
&\qquad \phi(\rho,P_{Y_{2}^{n}|V_{1}^{n},V_{2}^{n}=v_2, U^n
= u},\nonumber\\
&\qquad\sum_{\lambda}P_{\Lambda| V_{2}^{n} = v_2, U^n = u}P_{V_{1}^{n}|V_{2}^{n}=v_2,\Lambda=\lambda, U^n = u})\Big) \nonumber\\
&\qquad \textrm{ (by the concavity of $\exp(\phi)$
function)}\nonumber\\
&\leq 1+\mathbf{E}_{v_2, u}\exp\Big(n\rho(R_{\mathcal{I}}-R_p) +\nonumber\\
&\qquad \phi(\rho,P_{Y_{2}^{n}|V_{1}^{n},V_{2}^{n}=v_2, U^n =
u},P_{V_{1}^{n}|V_{2}^{n}=v_2, U^n = u})\Big) \label{belog}
\end{align}

Let us first focus on the nonconstant term of \eqref{belog}:
\begin{align}
&\mathbf{E}_{v_2, u}\exp\Big(n\rho(R_{\mathcal{I}}-R_p) +\nonumber\\
&\qquad \phi(\rho,P_{Y_{2}^{n}|V_{1}^{n},V_{2}^{n}=v_2, U^n =
u},P_{V_{1}^{n}|V_{2}^{n}=v_2, U^n = u})\Big) \nonumber\\
&= \prod_{i=1}^{n}\sum_{u_i\in \mathcal{U}, v_{2,i}\in \mathcal{V}_2}P_{U, V_2}(u_i, v_{2,i})\exp\Big(\rho(R_{\mathcal{I}}-R_p) +\nonumber\\
&\qquad \phi(\rho,P_{Y_{2}|V_{1},V_{2}=v_{2,i}, U =
u_i},P_{V_{1}|V_{2}=v_{2,i}, U = u_i})\Big) \nonumber\\
&= \Big[\sum_{u_i\in \mathcal{U}, v_{2,i}\in \mathcal{V}_2}P_{U, V_2}(u_i, v_{2,i})\exp\Big(\rho(R_{\mathcal{I}}-R_p) +\nonumber\\
&\qquad \phi(\rho,P_{Y_{2}|V_{1},V_{2}=v_{2,i}, U =
u_i},P_{V_{1}|V_{2}=v_{2,i}, U = u_i})\Big)\Big]^n \label{RHS}
\end{align}

Substitute the expression \eqref{RHS} back into \eqref{belog}, then
take $\log$ on both sides of the inequality \eqref{belog}, and use
the inequality $\log(x+1)\leq x,\ \forall x\geq 0$, then we can
obtain

\begin{align}
&I(\alpha_{1,\mathcal{I}}(F_{1}^{n}(C_{1}^{n}));Y_{2}^{n}|F_{1}^{n},\Lambda,
U^n)\nonumber\\
&\leq \frac{1}{\rho}\Big[\sum_{u_i\in \mathcal{U}, v_{2,i}\in \mathcal{V}_2}P_{U, V_2}(u_i, v_{2,i})\exp\Big(\rho(R_{\mathcal{I}}-R_p) +\nonumber\\
&\qquad \phi(\rho,P_{Y_{2}|V_{1},V_{2}=v_{2,i}, U =
u_i},P_{V_{1}|V_{2}=v_{2,i}, U = u_i})\Big)\Big]^n \label{fin}
\end{align}

Since what we concern is under what situation the above upper bound
goes to zero as $n\rightarrow \infty$, take the logarithm of
\eqref{fin} we have

\begin{eqnarray*}
-\log \rho +n\rho \Biggl[R_{\mathcal{I}}-R_p + A(\rho)\Biggr].
\end{eqnarray*}

In the above equation,

\begin{equation}% \label{eq:1}
\begin{split}
A(\rho) = &\frac{1}{\rho}\log \Big(\sum_{u_i, v_{2,i}{V}_2}P_{U,
V_2}(u_i, v_{2,i})\\
&\exp\big(\phi(\rho,P_{Y_{2}|V_{1},V_{2}=v_{2,i}, U =
u_i},P_{V_{1}|V_{2}=v_{2,i}, U = u_i})\big)\Big).
 \end{split}
 \end{equation}

We can see that $A(\rho) \rightarrow I(V_1; Y_2| U, V_2)$ as $\rho
\rightarrow 0$ by the l'H\^opital's rule.

%Even though the equivocation rate is the same as the one in our
%previous paper, but the achievable secure capacity region now is
%exactly the the Han-Kobayashi region.

Set the size of $\mathcal{C}_{1}^{n}$ as
\[
\frac{\log |\mathcal{C}_{1}^{n}|}{n} = R_p = R_{1}- \delta
\]

with $\delta>0$ such that
\begin{equation}
R_{\mathcal{I}} - R_{1,\mathcal{I},e}  >R_{\mathcal{I}} -
R_p+I(V_1;Y_2|UV_2)\label{eq:ratesetting}
\end{equation}

\noindent for all $\emptyset \neq \mathcal{I} \subseteq \{1$,
\ldots, $T\}$. Note that here $(R_{1}, R_{2})\in R_{HK}$.

Then by  Eq.\ (\ref{fin}), we can see that there exists $\epsilon_n
\rightarrow 0 (n\rightarrow \infty)$ such that
\begin{equation}
I(C_{1,\mathcal{I}};Y_{2}^{n}|F_{1}^{n},\Lambda, U^n) \leq
\epsilon_n \label{eq:gotozero}
\end{equation}
if $R_{\mathcal{I}} = R_{1,\mathcal{I},e}$. On the other hand, when
$R_{\mathcal{I}} > R_{1,\mathcal{I},e}$, by Eqs.\ \eqref{belog} and
\eqref{RHS}, we have
\begin{eqnarray}
&&\mathbf{E}_{f_{1}^{n},\lambda,u}\exp\Big(\rho I(C_{1,\mathcal{I}};Y_{2}^{n}|F_{1}^{n}=f_{1}^{n},\Lambda=\lambda, U^n = u)\Big)\nonumber\\
& \leq & 1+
\exp\Big(n\rho(R_{\mathcal{I}}-R_p+I(V_1;Y_2|UV_2)+\epsilon(\rho))\Big),
\label{eq:equivocation}
\end{eqnarray}
where $\epsilon(\rho) \rightarrow 0 (\rho\rightarrow 0)$. Let
$\delta_n$ be the decoding error probability of the underling
channel code for the interference channel. Then, by the almost same
argument as \cite{matsumotohayashi2011netcod}, there exists at least
one tuple $(f_{1}^{n},\lambda, u)$ such that

\begin{align}
&I(C_{1,\mathcal{I}};Y_{2}^{n}|F_{1}^{n},\Lambda, U^n) < 2 \times 2\times 2^T \epsilon_n \;(\textrm{if }R_{\mathcal{I}} = R_{1,\mathcal{I},e}),\nonumber\\
&\exp(\rho
I(C_{1,\mathcal{I}};Y_{2}^{n}|F_{1}^{n}=f_{1}^{n},\Lambda=\lambda,
U^n=u)) \nonumber\\
\leq& 2 \times 2\times 2^T \Big[1+
\exp\Big(n\rho(R_{\mathcal{I}}-R_p+
I(V_1;Y_2|UV_2)+\epsilon(\rho))\Big)\Big],\label{eq:equi2}\\
&\textrm{Decoding error probability} \leq 2 \times 2\times 2^T
\delta_n.\nonumber
 \end{align}

In the above expressions, $\epsilon(\rho)$ is a constant depends
only on $\rho$, and $\lim_{\rho\rightarrow 0}\epsilon(\rho) = 0$.

By Eq.\ (\ref{eq:equi2}) we can see
\begin{align}
&\frac{I(C_{1,\mathcal{I}};Y_{2}^{n}|F_{1}^{n}=f_{1}^{n},\Lambda=\lambda,
U^n=u))}{n} \nonumber\\
&\leq \frac{1+\log (2 \times 2\times 2^T)}{n\rho} +
R_{\mathcal{I}}-R_p +I(V_1;Y_2|UV_2)+\epsilon(\rho).\label{eq1000}
\end{align}

\noindent for $R_{\mathcal{I}}-R_p+I(V;Z|U)+\epsilon(\rho) \geq 0$,
where we used $\log(1+\exp(x)) \leq 1+x$ for $x\geq 0$. By Eq.\
(\ref{eq:ratesetting}) and (\ref{eq1000}) we can see that the
equivocation rate $H(C_{1,\mathcal{I}}|
Y_{2}^{n},F_{1}^{n}=f_{1}^{n},\Lambda=\lambda)/n$ becomes larger
than the required value $R_{1,\mathcal{I},e}$ for sufficiently large
$n$. This completes the analysis of the equivocation rates and the
mutual information for all $\emptyset \neq \mathcal{I} \subseteq
\{1$, \ldots, $T_1\}$.

\begin{remark}
In \eqref{eq1000}, we have ``almost'' provided a bound for the
mutual information even when the length of the codeword is finite.
We say ``almost'' because the error term $\epsilon(\rho)$ is not
explicitly determined. But we can always choose a small $\rho$, thus
make the error term as small as we want. In this way we bound the
equivocation rate in the worst case with finite codeword length.
\end{remark}
\end{proof}

\subsection{Outer Bound}

Next we will provide an outer bound for secure multiplex coding over
interference channels:

\begin{theorem}\label{outer}
An outer bound for the interference channels with secure multiplex
coding is as the following:

\begin{align*}
R_{1,\mathcal{I}_1, e}&\leq \min\left\{
      \begin{aligned}
         & I(V_1; Y_1|U) - I(V_1; Y_2|U) \\
         &I(V_1; Y_1|V_2, U) - I(V_1; Y_2|V_2,U)
      \end{aligned} \right\}\\
R_{2,\mathcal{I}_2, e}&\leq \min\left\{
      \begin{aligned}
         & I(V_2; Y_2|U) - I(V_2; Y_1|U) \\
         &I(V_2; Y_2|V_1, U) - I(V_2; Y_1|V_1,U)
      \end{aligned} \right\}\\
0 &\leq R_{1,\mathcal{I}_1, e}\leq \sum_{i\in\mathcal{I}_1}R_{1,i}\\
0 &\leq R_{2,\mathcal{I}_2, e}\leq \sum_{i\in\mathcal{I}_2}R_{2,i}\\
0 &\leq R_1 = \sum_{1\leq i\leq T_1}R_{1,i} \leq I(V_1; Y_1|U)\\
0 &\leq R_2 = \sum_{1\leq j\leq T_2}R_{2,j} \leq I(V_2; Y_2|U)\\
\end{align*}
\noindent In the above expressions, $\emptyset \neq
\mathcal{I}_1\subseteq \{1, \ldots, T_1\},\  \emptyset
\neq\mathcal{I}_2 \subseteq \{1, \ldots, T_2\}$ are subsets of the
index sets for transmitter 1 and 2 respectively.

Note that the variables in the above expressions satisfy the Markov
chain $U\rightarrow V_1V_2\rightarrow X_1X_2\rightarrow Y_1Y_2$.
%\noindent with variables $U, W_t, V_t, X_t\ (t = 1, 2)$ forming the
%Markov chains: $W_t\rightarrow V_t \rightarrow X_t$ and
%$U\rightarrow V_t \rightarrow X_t$.
\end{theorem}

\begin{proof}
The outer bound is constructed by using the techniques in
\cite{Interference} section IV, by combining two bounds obtained
from different set of inequalities.

Obviously it can be seen,

$$R_t \leq I(V_t; Y_t|U)$$

Since after receiver $t$ receives $Y_{t}^{n}$, it will be able to
decode the confidential message $C_{t}^{n}$ with high probability
(reliable transmission requirement), we can express this as:

\addtocounter{equation}{1}
\begin{align}
H(C_1|Y_1) \leq n\delta_1 \tag{\theequation a}\\
H(C_2|Y_2) \leq n\delta_2 \tag{\theequation b}
\end{align}
%\noindent Note that the above inequality actually means that each of
%the above terms is small enough.

With out loss of generality, we only consider for the case of $R_{1,
\mathcal{I}_1, e}$.

\subsubsection{Construct the 1st bound}

\begin{equation} \label{eq:1}
R_{1, \mathcal{I}_1, e}\leq H(C_1|Y_2)\leq H(C_1|Y_2) - H(C_1|Y_1)+
n\delta_1
 \end{equation}

Then by almost the same arguments as in \cite{Interference} section
IV part A (just change $W$ to $C$), we have

\begin{equation} \label{equivo_1}
R_{1, \mathcal{I}_1, e}\leq I(V_1; Y_1|U) - I(V_1; Y_2|U)
 \end{equation}

\subsubsection{Construct the 2nd bound}
\begin{equation} \label{eq:2}
\begin{split}
&R_{1, \mathcal{I}_1, e} \leq H(C_1|Y_2)\leq H(C_1, C_2|Y_2)\\
&= H(C_1|Y_2,  C_2) + H(C_2|Y_2)\\
&\leq H(C_1|Y_2,  C_2) + n\delta_2 - H(C_1|Y_1) + n\delta_1\\
&\leq H(C_1|Y_2,  C_2) - H(C_1|Y_1,  C_2) + n(\delta_1+\delta_2)
 \end{split}
 \end{equation}

By the same argument in \cite{Interference} section IV part B, we
have

\begin{equation} \label{equivo_2}
R_{1, \mathcal{I}_1, e}\leq I(V_1; Y_1|V_2, U) - I(V_1; Y_2|V_2, U)
 \end{equation}
\end{proof}

\begin{remark}
If we look carefully into the outer and inner bounds, we can see
that when $R_2 = R_{2, e} = 0$, which means the model degenerates to
the wiretap channel, the inner and outer bound coincides and become
the capacity-equivocation region of the wiretap channel.
\end{remark}

%%%%%%%%%%%%%%%%%%%%%%%%%%%%%%%%%%%%%%%%%%%%%%%%%%%%%%%%%%%%%%%%%%%%%%%555
% Formula Page
%\newcounter{TempEqCnt}
%\setcounter{TempEqCnt}{\value{equation}}
%\setcounter{equation}{x}

\begin{figure*}[ht]
\begin{align}
I(V_1; Y_1| W_1W_2) &=
\frac{1}{2}\log\Biggl[1+\frac{\beta_1\theta_1(1-\mu_1)P_1}{1+
\beta_1(1-\theta_1)P_1+
\tau_{1}^{2}\beta_2(1-\theta_2\mu_2)P_2}\Biggr]\label{V_1; Y_1| W_1W_2}\\
I(W_2V_1; Y_1| W_1) &=
\frac{1}{2}\log\Biggl[1+\frac{\beta_1\theta_1(1-\mu_1)P_1
+\tau_{1}^{2}\beta_2\theta_2\mu_2P_2 }{1+ \beta_1(1-\theta_1)P_1+
\tau_{1}^{2}\beta_2(1-\theta_2\mu_2)P_2}\Biggr]\label{W_2V_1; Y_1| W_1}\\
I(V_1; Y_1| W_2) &= \frac{1}{2}\log\Biggl[1+\frac{\beta_1\theta_1P_1
}{1+ \beta_1(1-\theta_1)P_1+
\tau_{1}^{2}\beta_2(1-\theta_2\mu_2)P_2}\Biggr]\label{V_1; Y_1| W_2}\\
I(W_2V_1; Y_1) &=
\frac{1}{2}\log\Biggl[1+\frac{\beta_1\theta_1P_1+\tau_{1}^{2}\beta_2\theta_2\mu_2P_2
}{1+ \beta_1(1-\theta_1)P_1+
\tau_{1}^{2}\beta_2(1-\theta_2\mu_2)P_2}\Biggr]\label{W_2V_1; Y_1}\\
I(V_1; Y_2|V_2) &=
\frac{1}{2}\log\Biggl[1+\frac{\tau_{2}^{2}\beta_1\theta_1P_1}{1+
\tau_{2}^{2}\beta_1(1-\theta_1)P_1+ \beta_2(1-\theta_2)P_2
}\Biggr]\label{V_1; Y_2|V_2}\\
I(V_1;Y_1) &= \frac{1}{2}\log\Biggl[1+\frac{\beta_1\theta_1P_1 }{1+
\beta_1(1-\theta_1)P_1+
\tau_{1}^{2}\beta_2P_2}\Biggr]\label{V_1; Y_1}\\
I(V_1;Y_2) &=
\frac{1}{2}\log\Biggl[1+\frac{\tau_{2}^{2}\beta_1\theta_1P_1
}{1+\tau_{2}^{2} \beta_1(1-\theta_1)P_1+
\beta_2P_2}\Biggr]\label{V_1; Y_2}\\
I(V_1;Y_1|V_2) &= \frac{1}{2}\log\Biggl[1+\frac{\beta_1\theta_1P_1
}{1+\beta_1(1-\theta_1)P_1+
\tau_{1}^{2}\beta_2(1-\theta_2)P_2}\Biggr]\label{V_1; Y_1|V_2}
\end{align}
\vspace*{10pt} \hrulefill
\end{figure*}

%\setcounter{equation}{\value{TempEqCnt}}
%%%%%%%%%%%%%%%%%%%%%%%%%%%%%%%%%%%%%%%%%%%%%%%%%%%%%%%%%%

\section{Gaussian Interference Channel}

In this section we consider Gaussian interference channel with
confidential messages. As in \cite{interference_model}, the channel
input and output are real numbers, and the channel in Fig.
\ref{Interf} is specified as:

\addtocounter{equation}{1}
\begin{align}\label{GIC}
Y_1 = X_1 + \tau_1 X_2 + N_1 \tag{\theequation a}\\
Y_2 = \tau_2 X_1 +  X_2 + N_2 \tag{\theequation b}
\end{align}

\noindent $\tau_1$ and $\tau_2$ are normalized crossover
\emph{channel gains}, $X_t$ has the average power constraint:

$$\sum_{i=1}^{n}\frac{E[X_{t, i}^{2}]}{n}\leq P_t.$$

\noindent and $N_1$ and $N_2$ are two independent, zero-mean, unit
variance, Gaussian noise variables. %We focus on the so called weak
%interference channel, that is $0\leq \tau_{1}^{2}\leq 1$, $0\leq
%\tau_{2}^{2}\leq 1$.

We can easily carry over our proof in the last section to the case
of Gaussian channel, because

\begin{itemize}
\item By replacing the probability mass functions $P_{Z|L}$ and $P_Z$ by
their probability density functions, Eq.\ (\ref{hpa1discrete})  can
be extended to the Gaussian case.

\item The random codebook $\Lambda$ obeys
multidimensional Gaussian distribution.

\item The concavity of $\exp(\phi)$ is preserved even after
its second argument is changed to be conditional probability
density.

\item All the derivations in the last
section hold true even when $V_t$, $Y_t$, $\Lambda$ are continuous
and their probability mass functions are replaced with probability
density functions, while $U$, $B_{t}^{n}$, $F_{t}^{n}$ remain to be
discrete random variables over finite alphabets.
\end{itemize}

To get the expression for the inner bound on the capacity region for
secure multiplex coding, we only need to calculate for $I(V_1; Y_1|
W_1W_2)$, $I(W_2V_1; Y_1| W_1)$, $I(V_1; Y_1| W_2)$, $I(W_2V_1;
Y_1)$, $I(V_1; Y_2|V_2)$, $I(V_1;Y_1)$, $I(V_1;Y_2)$ and
$I(V_1;Y_2|V_2)$ because of the symmetry of the coefficients. (For
other expressions, we only need to switch $1$ and $2$.)

The scheme in our consideration is similar to \cite{Interference}
except that we allow both transmitters to generate artificial noise:

Assume transmitter $t$ only use a fraction of $\beta_t$ of its
maximum power ($0\leq \beta_t\leq 1$). Among the transmission power,
transmitter $t$ then takes out a fraction of $(1- \theta_t)$ ($0\leq
\theta_t\leq 1$) to make artificial noise to achieve the secrecy
requirement. Among the power devoted for transmitting confidential
message, a fraction of $\mu_t $ is used over the ``common channel''.
Let $U$ serve as a convex combination operator, thus we have (for $t
= 1, 2$)

$$X_t = V_t + A_t,\ V_t = W_t + Q_t$$

\noindent where $W_t, Q_t, A_t$ are independent Gaussian random
variables with $W_t \sim  \mathcal{N}[0, \beta_t\mu_t\theta_tP_t]$,
$Q_t \sim \mathcal{N}[0, \beta_t(1-\mu_t)\theta_tP_t]$ and $A_t \sim
\mathcal{N}[0, \beta_t(1-\theta_t)P_t]$.

It is straightforward to evaluate the mutual information mentioned,
and we have \eqref{V_1; Y_1| W_1W_2} $\sim$ \eqref{V_1; Y_2|V_2}.

Take all the above equations into theorem \ref{inner} and
\ref{outer}, we have obtained the expressions of inner and outer
bounds over the Gaussian interference channel.

\section{Discussion and Comparison}

Compare with the results in \cite{Interference} where also the inner
and outer bounds are given, one of the most obvious advantage of our
results is that by secure multiplex coding, the rate loss incurred
by adding dummy message to achieve security is removed, thus the
maximum total transmission rate is increased.

Now we compare our results with that in \cite{Interference} when
$T_1 = T_2 = 1$, which means that  a ``trivial'' multiplex coding
with only one channel is used. This scheme is equivalent to adding
dummy message to achieve secrecy as in \cite{Interference}.

For convenience of comparison, we denote the rate of confidential
messages transmitted by transmitter $t$ is $R_t$, while the rate of
the dummy message is denoted by $R_{t}^{'}$. And the confidential
messages is conveyed by random variable $C_1$ and $C_2$ for the two
transmitters respectively.

We can see the outer bound becomes the same as in
\cite{Interference}, %while the inner bound becomes
%
%\begin{align*}
%0\leq R_{1}&\leq R_{1}^{\circ} - I(V_1;Y_2|U, V_2)\\
%0\leq R_{2}&\leq R_{2}^{\circ} - I(V_2;Y_1|U, V_1)
%\end{align*}
%
%Remember that $(R_{1}^{\circ}, R_{2}^{\circ})\in R_{CMG}$ denotes
%rate pair that is on the boundary of the Han-Kobayashi region.
so in the following we focus on the inner bound.

If we let $B_t = C_t$ and $E_t = \emptyset$ this means we remove the
variables $W_1$ and $W_2$ from \eqref{Han-Kobayashi_1} to
\eqref{Han-Kobayashi_final}, then we can see that the capacity
region in theorem \ref{inner} becomes

\begin{equation} \label{Inner_Inter}
\begin{split}
0\leq R_{1}&\leq I(V_1;Y_1|U) - I(V_1;Y_2|U, V_2)\\
0\leq R_{2}&\leq I(V_2;Y_2|U)- I(V_2;Y_1|U, V_1)
 \end{split}
 \end{equation}

\noindent which is exactly the inner bound proposed in
\cite{Interference}. In this way we can easily see that at least our
proposed inner bound is always not worse than that proposed in
\cite{Interference}. %And in some cases where it is necessary to send
%on both "common" and "private" channel to achieve the Han-Kobayashi
%region, our proposed bound is strictly better than the one proposed
%in \cite{Interference}.

Moreover by \eqref{eq1000}, we can bound the amount of leakage of
confidential message for the unintended receiver (without loss of
generality we consider for $C_1$):

\begin{align}\label{Upp_Com}
&\frac{I(C_{1};Y_{2}^{n}|F_{1}^{n}=f_{1}^{n},\Lambda=\lambda,
U^n=u))}{n} \nonumber\\
&\leq \frac{1+\log (2 \times 2\times 2^T)}{n\rho} - R_{1}^{'}
+I(V_1;Y_2|UV_2)+\epsilon(\rho).
\end{align}

Note \eqref{Upp_Com} is valid even when the rate of the secret
transmission is higher than the described secret capacity.

In the following, we will compare the inner bound proposed in
theorem \ref{inner} and \eqref{Inner_Inter} in \cite{Interference},
and by show some numerical results in several Gaussian interference
channels (different $\tau_1$ and $\tau_2$) with different power
constraints ($P_1$ and $P_2$).

The ``Secure Region in \cite{Interference}'' indicates
\eqref{Inner_Inter}, while the ``Secure Han-Kobayashi Region'' is
presented in theorem \ref{inner}. %These results coincide with those
%presented in Fig.5 in \cite{Interference}. The difference is that in
%their plots they use log with base 2, while we use natural log.
In the cases where $P_1 = P_2$ and $\tau_1 = \tau_2$, we are not
able to find any visible difference between the two bounds, as shown
in Fig. \ref{P10_t0.2}. But when the power constraints are different
for the two transmitters (i.e. $P_1\neq P_2$) and the channel is not
symmetric (i.e. $\tau_1\neq \tau_2$), we find some cases where our
proposed inner bound is strictly larger as shown in Fig.
\ref{P80_P10_01_03}.

\begin{figure}[t]
\centering{}\includegraphics[scale=0.6,trim=0 0 0
0,clip=]{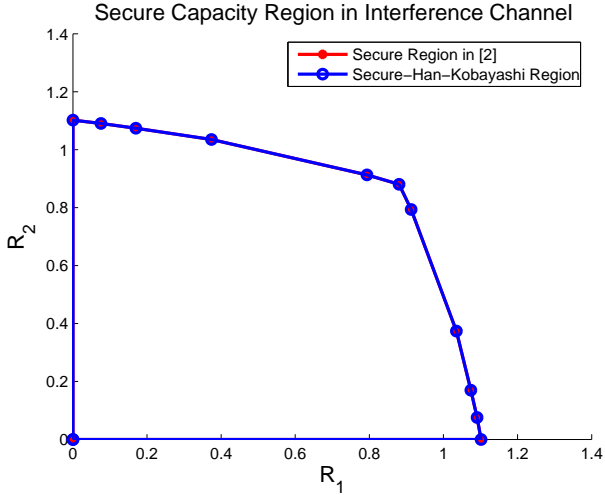}\caption{Achievable regions with
$\tau_1=\tau_2 =0.2$ and $P_1 = P_2 = 10$.} \label{P10_t0.2}
\end{figure}

\begin{figure}[t]
\centering{}\includegraphics[scale=0.6,trim=0 0 0
0,clip=]{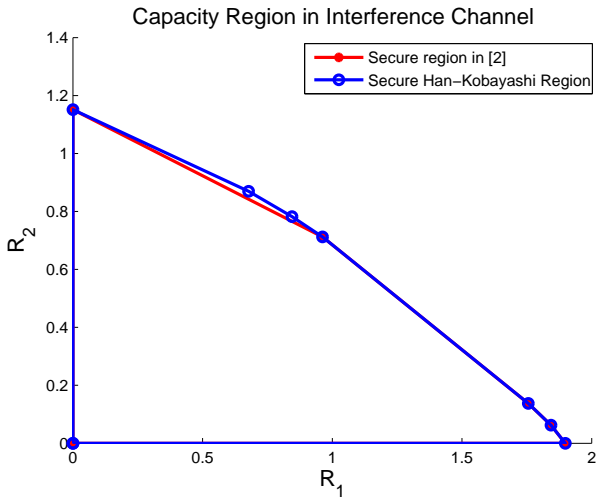}\caption{Achievable regions with $\tau_1=
0.1, \tau_2 =0.3$ and $P_1 =80, P_2 = 10$.} \label{P80_P10_01_03}
\end{figure}

Here we also plot the best achievable secure region of the two
schemes, by this we mean that we add ``non-trivial'' multiplex
coding to remove the rate loss caused by the dummy messages. In
\cite{Interference}, the region is expressed as

\begin{equation} \label{Inter_best}
\begin{split}
0\leq R_{1}&\leq I(V_1;Y_1|U)\\
0\leq R_{2}&\leq I(V_2;Y_2|U)
 \end{split}
 \end{equation}

\noindent while our proposed scheme can achieve the whole
Han-Kobayashi region. So we also compare \eqref{Inter_best} with the
Han-Kobayashi region in Fig. \ref{P10_t0.5} and \ref{P100_t0.5}.
These plots displayed oblivious improvement which comes from the
coding scheme that we applied.

\begin{figure}[t]
\centering{}\includegraphics[scale=0.6,trim=0 0 0
0,clip=]{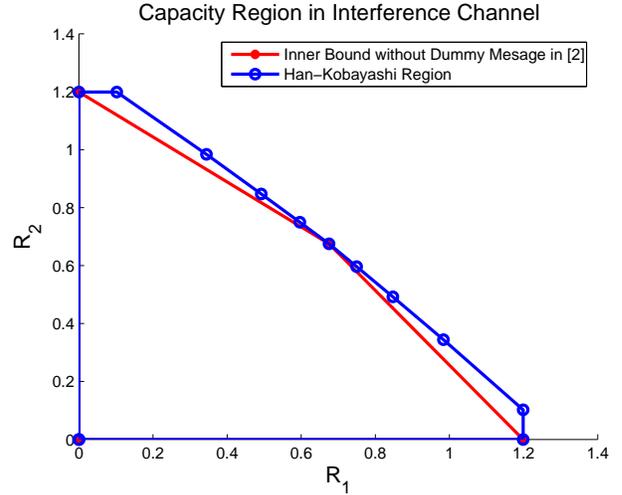}\caption{Achievable regions with
$\tau_1=\tau_2 =0.5$ and $P_1 = P_2 = 10$.} \label{P10_t0.5}
\end{figure}

\begin{figure}[t]
\centering{}\includegraphics[scale=0.6,trim=0 0 0
0,clip=]{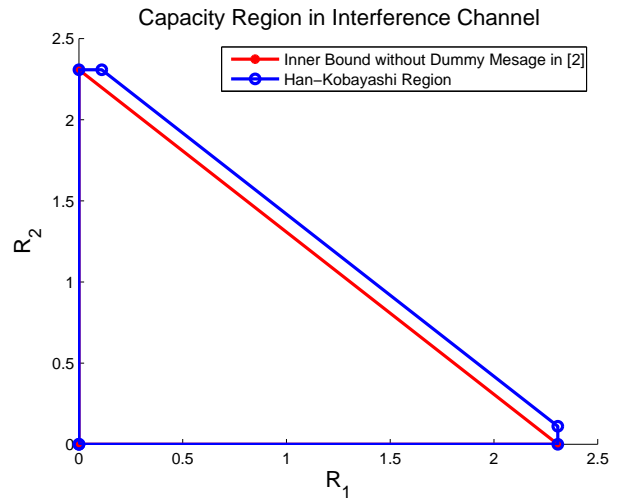}\caption{Achievable regions with
$\tau_1=\tau_2 =0.5$ and $P_1 = P_2 = 100$.} \label{P100_t0.5}
\end{figure}

\section{Conclusions}
In this paper inner and outer bounds for secure multiplex coding
over interference channel has been proposed, we have also presented
a random coding scheme that achieves the inner bound. We improved
the inner bound in \cite{Interference} and pushed the inner bound to
the Han-Kobayashi region. Also we have substituted the weak secrecy
requirement by the strong one, and removed the information rate loss
caused by the dummy message. Moreover we evaluated the equivocation
rate for a collection of secret messages. Finally we extended our
results to the case of Gaussian channel.

\section{Acknowledgement}
This work was partly supported by the AoE Grant E-02/08 from the
University Grants Committee of the Hong Kong Special Administration
Region, China, and also by the MEXT Grant-in-Aid for Young
Scientists (A) No.\ 20686026 and (B) No.\ 22760267, and Grant-in-Aid
for Scientific Research (A) No.\ 23246071.

\end{document}